\documentclass[a4paper,onecolumn,11pt]{article}

\usepackage{setspace}
\usepackage{appendix}
\usepackage{titlesec}
\usepackage[normalem]{ulem}
\usepackage{bm}
\titleformat*{\section}{\center \bf}
\titleformat*{\subsection}{\raggedright \bf}

\usepackage[margin=0.8in]{geometry}
\usepackage{indentfirst}
\usepackage{graphicx}
\usepackage{ dsfont }
\usepackage{amsmath}
\usepackage{amssymb}
\usepackage{amsopn}
\usepackage{float}
\usepackage{stfloats}
\usepackage{cite}
\usepackage{amsthm}

\usepackage{xcolor}

\newtheorem{theorem}{Theorem}

\title{\bf{Compute-and-Forward Two-Way Relaying}}

\author{
 \small Seyed Mohammad Azimi-Abarghouyi\\
\small \texttt{sm$\_$azimi@ee.sharif.edu}\\
\and
\small Mohsen Hejazi\\
\small \texttt{mhejazi@ee.sharif.edu}\\
\and
 \small Masoumeh Nasiri-Kenari\\
\small \texttt{mnasiri@sharif.edu}\\
\and
 \small Wireless Research Laboratory (WRL)\\ \small Electrical Engineering Department, Sharif University of Technology, Tehran, Iran}

\date{}

\begin{document}
\maketitle

\begin{abstract}
\textbf{
In this paper, a new two-way relaying scheme based on compute-and-forward (CMF) framework and relay selection strategies is proposed, which provides a higher throughput than the conventional two-way relaying schemes. Two cases of relays with or without feedback transmission capability are considered. An upper bound on the computation rate of each relay is derived, and based on that, a lower bound on the outage probability of the system is presented assuming block Rayleigh fading channels. Numerical results show that while the average sum rate of the system without feedback, named as Max Compute-and-Forward (M-CMF), reaches the derived upper bound only in low SNRs, that of the system with feedback, named as Aligned Compute-and-Forward (A-CMF) reaches the bound in all SNRs. However, both schemes approach the derived lower bound on the outage probability in all SNRs. For the A-CMF, another power assignment based on applying the constraint on the total powers of both users rather than on the power of each separately, is introduced. The result shows that the A-CMF performs better under the new constraint. Moreover, the numerical results show that the outage performance, average sum rate, and symbol error rate of the proposed schemes are significantly better than those of two-step and three-step decode-and-forward (DF) and amplify-and-forward (AF) strategies for the examples considered.}

\par
\textbf{
Index Terms- compute and forward, max compute-and-forward, aligned compute-and-forward, feedback, two-way relaying, relay selection, outage probability, average sum rate, symbol error rate.}
\end{abstract}


\section{INTRODUCTION} 
Two way relaying communications have recently attracted considerable attentions due to their various applications. In this communication scenario, two users attempt to communicate with each other with the help of a relay. To this end, physical layer network coding (PLNC) [1] along with the conventional DF or AF relaying strategy has been commonly considered [2-4]. It has been shown that PLNC can achieve within 1/2 bit of the capacity of a Gaussian TWRC (Two Way Relay Channel) and this is asymptotically optimal at high SNRs [5-6]. In [2-3], based on DF startegy, two-step and three-step two-way relaying schemes are proposed. In the two-step scheme, in the first step, both users simultaneously transmit their messages, and the relay recovers both messages in turn, using a linear receiver structure like successive interference cancellation (SIC) [7]. In the second step, the relay sends a combination of the recovered messages to the users. The problem with the scheme proposed is that, when recovering one of the messages, the other message is considered as noise, which results in a performance loss. As a solution, an optimum ML decoder can be utilized at the relay at the expense of a very high complexity [3]. The three-step DF two-way relaying proposed in [2] requires three time slots that results in a throughput reduction. As an alternative, the relay can exploit AF strategy to simply amplify the received signal from the users, and then forward it to the users. Due to noise amplification in the relay, this scheme shows a poor performance [2].

The novel relaying strategy known as compute-and-forward (CMF), proposed by Nazer and Gastpar [8], is proved to be efficient for multiuser communication scenarios. The CMF strategy can exploit the interference to achieve a higher throughput. CMF strategy is also known as a reliable physical layer network coding [9]. In CMF strategy, all sources transmit simultaneously. Each relay, based on its received signal (a noisy and channel weighted combination of the users' codewords) and its knowledge of the channel coefficients, decodes an equation, which is an integer-linear combination of the users' transmitted messages. The integer coefficients of the equation are presented by a vector called an equation coefficient vector (ECV). The relay has to find the ECV with the highest possible rate. The relay then transmits the decoded equation to the destination. The destination recovers the desired messages by receiving sufficient number of decoded equations from the relays. For codewords, lattice codes are commonly utilized, which can achieve the capacity of additive white Gaussian noise (AWGN) channels [10]. While CMF strategy has been considered in different scenarios in the literature, such as multi-antenna systems [11], cooperative distributed antenna systems [12], multi-access relay channels [13], generalized multi-way relay channels [14], two transmitter multi-relay systems [15], and finally multi-source multi-relay network [16]; however, to our best knowledge the application of CMF in two-way relaying hasn't been considered so far, just from information theory aspect in [17]. In this paper, we propose a new practical framework for two-way relaying based on CMF strategy, in which we use a linear receiver and a general lattice encoding previously proposed by Nazer [8]. 

 We consider this framework for two cases. First, we investigate the relays without the capability of sending any feedback to the users. We call the corresponding proposed scheme as max compute-and-forward (M-CMF). Then we consider the relays that have feedback capability; the related scheme is called as aligned compute-and-forward (A-CMF). For the latter case, the power can be efficiently allocated to the users in a way to increase the computation rate through aligning the scaled channels to the integer coefficients, under a maximum power constraint for each user. The proposed schemes, in contrast to DF and AF based schemes, can handle both the interference and noise, and thus enhance the network throughput considerably. To achieve a higher order of diversity, multiple relays along with a simple relay selection technique are employed. We consider a block Rayleigh fading channels between the users and the relays. The channels have phase variations in addition to the amplitude variations. While the proposed schemes have been considered and work quite well for general complex Gaussian channels with variation in both phase and amplitude, for the sake of simplicity and tractability of the analytical performance evaluation, in this paper, we focus on addressing the amplitude variation of the channels and do not consider the carrier phase offset for the analytical performance analysis. In the other words, we assume that the phase offset between two received users' signals has been compensated at the best relay. This makes the channels realized by the best relay be real-valued Rayleigh channels. This assumption was commonly used in the literature when considering the performance analysis of CMF based strategies, for instance please see [15] and [18]. However for simulation evaluations, we consider general Rayliegh fading channel (Complex Gaussian Coefficient) with both phase and amplitude variations.  We analitically derive, an upper bound on the computation rate of each relay. Then using this bound, we derive a bound on the average sum rate and the outage probability of the proposed scheme. Based on the bound obtained for the outage probability, we derive the system diversity order. Numerical results show that A-CMF reaches the bound in all SNRs, M-CMF is tight on the bound only in low SNRs; in the other words, A-CMF improves the compute rate at high SNR. To have a fair comparison with M-CMF and the other schemes, A-CMF under a power constraint on the total powers of both users rather than on the power of each user is also considered. As expected, the numerical results verifies that A-CMF under the new constraint performs better. For the latter case, A-CMF and M-CMF are in fact compared under the same total system power. We also evaluate the symbol rate of the system, and compare the results with those of the conventional schemes, which indicates the substantial superiority of the proposed schemes.

The remainder of this paper is organized as follows. In Section II, the system model is described. Section III presents the proposed method and the performance analysis is given in Section IV. Numerical results are presented in Section V. Finally, Section VI concludes the paper.

\section{SYSTEM MODEL} 
We consider a two way relay channel with two users and $M$ relays, as shown in Fig. 1. User $j, j=1,2$, exploits a lattice encoder with power constraint $ \alpha _j^2$ to project its message $w_j$ to a length-n complex-valued codeword $x_j$ such that ${\left| {\left| {{x_j}} \right|} \right|^2} \le n{\alpha _j^2}$. ${P_j}$  is considered as the maximum power constraint of the user $j$ ($\alpha _j^2 \le {P_j}$). We assume that each relay has a power constraint equal to $P_r$ that is more than or equal to $\max(P_1,P_2)$. The channel coefficient from user $j, j=1,2$, to relay $i, i = 1,2, \ldots ,M,$ denoted by  $h_{ji}$ and assumed to be equal to the reverse link coefficient $h_{ij}$, follows a real-valued Rayleigh distribution with variance $\sigma _{ji}^2$\footnote{As stated in the introduction, for the simplicity of the presentation and tractability of the analytical analysis, like in [18], here we only focus on the the real channel. For the numerical part, we also consider more general complex Gaussian channel.}. All channel coefficients for different $i$ and $j$ are assumed to be independent. We assume a block fading such that the coefficients remain constant during total transmission time slots required for the message exchanges. There is no direct link between two users. The noise received at the $i$'th relay and at the $j$'th user, denoted by $z_i^r$ and $z_j$ respectively, are i.i.d. according to the zero-mean Gaussian distribution with variance 1.

\section{COMPUTE AND FORWARD TWO WAY RELAYING}
In our proposed CMF based scheme, Multiple Access Broadcast (MABC) protocol for two way transmission is used [19]. In the first time slot (named multiple access phase), two users simultaneously transmit their codewords to the relays. Each relay receives a noisy linear combination of users' codewords, and using the CMF strategy [8], the relay decodes an equation, i.e., an integer linear combination of users' messages (see Subsection A). Then in the second time slot (broadcast phase), the best relay is selected (see Subsection C) to transmit its decoded equation to the users. Finally, by receiving the equation, each user recovers the other user's message (see Subsection B).

\subsection{Computation of the integer equation}
The received signal in each relay $i$ in the multiple access phase can be written as
\begin{eqnarray}
y_i^r = {h_{1i}}{x_1} + {h_{2i}}{x_2} + z_i^r , i = 1,2, \ldots ,M
\end{eqnarray}
\\For each relay $i$, vector ${\mathbf{h}_i}$ and matrix ${\mathbf{H}_i}$ are defined as
\begin{eqnarray}
\mathbf{h}_i \buildrel \Delta \over= \left[ {\begin{array}{*{20}{c}}{{\alpha _1}{h_{1i}}}\\{{\alpha _2}{h_{2i}}}\end{array}} \right]
\end{eqnarray}
and,
\begin{eqnarray}
\mathbf{H}_i \buildrel \Delta \over =  \mathbf{I} - \frac{{{\mathbf{h}_i}{\mathbf{h}_i}^T}}{1 + {{|| \mathbf{h}_i ||}^2}}
\end{eqnarray}

Based on CMF strategy [8], each relay has to decode the Equation Coefficient Vector (ECV), i.e. ${\mathbf{a}_i}= {\left[ {\begin{array}{*{20}{c}}{{a_{1i}}}&{{a_{2i}}}\end{array}} \right]^T} \in {Z^2}$, such a way to maximize its computation rate, i.e., the rate of recovering the equation, 
\begin{eqnarray}
{\mathbf{a}_i} &=& arg\mathop {\max }\limits_{\mathbf{a} \in {Z^2},\mathbf{a} \ne 0} {\log ^ + }\left( {{{\left( {{{\left| {\left| \mathbf{a} \right|} \right|}^2} - \frac{{{{\left| {{\mathbf{h}_i}^T\mathbf{a}} \right|}^2}}}{{1 + {{\left| {\left| {{\mathbf{h}_i}} \right|} \right|}^2}}}} \right)}^{ - 1}}} \right)\nonumber\\&=& arg\mathop {\min }\limits_{\mathbf{a} \in {Z^2},{\mathbf{a} \ne 0}} \left( {{\mathbf{a}^T}{\mathbf{H}_i}\mathbf{a}} \right)
\end{eqnarray}
where ${\rm{lo}}{{\rm{g}}^ + }\left( {\rm{x}} \right) = {\rm{max}}\left( {\log \left( {\rm{x}} \right),0} \right)$.
The solution of (4) can be expressed [20] as scaling the received signal by the following factor
\begin{eqnarray}
{\beta _i} = \frac{{\mathbf{h}_i}^T\mathbf{a}}{{1 + {{\left| {\left| {{{\mathbf{h}}_i}} \right|} \right|}^2}}}
\end{eqnarray}
That is the recovered equation is expressed as
\begin{eqnarray}
{s_i} = Q({\beta _i}y_i^r) = {a_{i1}}{x_1} + {a_{i2}}{x_2}
\end{eqnarray}
where $Q(.)$ shows the lattice quantizer function.
The computation rate $R_i^r$ of the equation $s_i$ is given by [8]
\begin{eqnarray}
R_i^r &=&{\rm{lo}}{{\rm{g}}^ + }\left( {\frac{1}{{1 + {{\left| {\left| {{\beta _i}{\mathbf{h}_i} - {\mathbf{a}_i}} \right|} \right|}^2}}}} \right)\nonumber\\&=&{\rm{lo}}{{\rm{g}}^ + }\left( {{{\left( {{\mathbf{a}_i}^T{\mathbf{H}_i}{\mathbf{a}_i}} \right)}^{ - 1}}} \right)
\end{eqnarray}

To enable each user to recover the other user’s message, the selected equation must not contain zero components. Hence, ${\mathbf{e}_1} = {\left[ {\begin{array}{*{20}{c}}1&0\end{array}} \right]^T}$ and  ${\mathbf{e}_2} = {\left[ {\begin{array}{*{20}{c}}0&1\end{array}} \right]^T}$ , as a solution of (4), are not desirable. Please note that, from [15, lemma 3], the only vectors that have zero elements and can be the solution of (4) are ${\mathbf{e}_1}$ and ${\mathbf{e}_2}$.

If the vector ${\mathbf{a}_i}$ computed from (4) has nonzero components, it will be chosen as the coefficient vector of the selected equation.
Otherwise, we do as follows. From (6), the relay first recovers the $k$'th message
\begin{eqnarray}
{s_i} = {x_k}
\end{eqnarray}
Then, by removing the effect of the codeword $x_k$ from the received vector $y_i^r$, the other user codeword, ${x_{k'}}, k' \ne k$, is also recovered. Finally, The relay constructs the equation ${s_i} = {x_1} + {x_2}$ with nonzero coefficients for the transmission.\\
A necessary condition for the occurrence of ${\mathbf{e}_k},k=1,2$, as the solution of (4) is
\begin{eqnarray}
{\alpha _k^2}h_{ki}^2 = {\rm{max}}\{ {\alpha _1^2}h_{1i}^2,{\alpha _2^2}h_{2i}^2\} 
\end{eqnarray}
From (7), the computation rate of ${\mathbf{e}_k},k=1,2$, can be computed as
\begin{eqnarray}
{R_e^i} = {\rm{log}}\left( {1 + \frac{{{\rm{max}}\{ {\alpha _1^2}h_{1i}^2,{\alpha _2^2}h_{2i}^2\} }}{{1 + {\rm{min}}\{ {\alpha _1^2}h_{1i}^2,{\alpha _2^2}h_{2i}^2\} }}} \right)
\end{eqnarray}
By removing $x_k$ from $y_i^r$, the rate of recovering the other message in the relay $i$ can be easily found as
\begin{eqnarray}
R_{{e^c}}^i= \log \left( {1 + {\rm{min}}\{ {\alpha _1^2}h_{1i}^2,{\alpha _2^2}h_{2i}^2\} } \right)
\end{eqnarray}
Hence, the rate of constructing the equation $s_i=x_1+x_2$ in each relay $i$ is given by
\begin{eqnarray}
R_i^r = {\rm{min}}( {R_e^i},{R_{{e^c}}^i} )
\end{eqnarray}

For users' power allocation, we consider two cases based on whether or not the relay is able to send some information feedback to the users, as follows:

In the case that the relays haven't feedback capability, the user $k$  transmits with the maximum possible power ${P_k}$ ($\alpha _k^2 = {P_k}$). We call this case as Max Compute and Forward (M-CMF).

In the another case when there is feedback capability for the relays, we propose Aligned Compute and Forward (A-CMF) scheme, as follows.
In this scheme, a new power adaptation algorithm is exploited, in which the users adjust their transmitted powers in order to minimize the quantization noise ${\varepsilon _m}$ between the scaled received signal at the best relay $m$ and its computed equation and hereby to increase the computation rate.  

We define a power adaptation vector (PAV) $\bm{\alpha}$ as
\begin{eqnarray}
\bm{\alpha} \buildrel \Delta \over = \left[ \alpha _1^2,\alpha _2^2 \right]
\end{eqnarray}
Based on the computation rate in (7), in each relay $i$, for the case ${\mathbf{a}_i} \ne {\mathbf{e}_1}$ or ${\mathbf{a}_i} \ne {\mathbf{e}_2}$, the quantization noise $\epsilon_i$ can be written as
\begin{eqnarray}
{\varepsilon _i} &=& {\left| {\left| {{\beta _i}{\mathbf{h}_i} - {\mathbf{a}_i}} \right|} \right|^2} \nonumber\\
&=& {\left( {{\beta _i}{\alpha _1}{h_{1i}} - {a_{1i}}} \right)^2} + {\left( {{\beta _i}{\alpha _2}{h_{2i}} - {a_{2i}}} \right)^2}
\end{eqnarray}
where

\hspace*{+30pt}$\alpha _1^2 \le {P_1}$ and $\alpha _2^2 \le {P_2}$.
\vspace{+10pt}

The ECV $\mathbf{a}_i$, PAV $\bm{\alpha}$, and scaling factor $\beta_{i}$ are the unknown parameters and should be optimally selected based on the maximization of the computation rate of the relay. We consider an alternative manner to solve this optimization problem. In the first step, assuming the PAV is known, the ECV and the scaling factor are computed from (4) and (5). In the second step, similarly, by assuming ECV and scaling factor are known based on using the values computed at the previous step, the PAV is calculated as follows:

Using KKT conditions by taking the derivation of (14) and putting the result equal to zero, we have
\begin{eqnarray}
\frac{{\partial {\varepsilon _i}}}{{\partial {\alpha _k}}} + {\mu _{ki}}\frac{{\partial \alpha _k^2}}{{\partial {\alpha _k}}} = 0,\forall k = 1,2
\end{eqnarray}
where ${\mu _{ki}}$ is the KKT coefficient related to the user $k$. (15) leads to:
\begin{eqnarray}
2{\beta _i}{h_{ki}}\left( {{\beta _i}{\alpha _k}{h_{ki}} - {a_{ki}}} \right) + 2{\alpha _k}{\mu _{ki}} = 0 ,\forall k = 1,2
\end{eqnarray}
Hence, we have
\begin{eqnarray}
{\alpha _k} = \frac{{{\beta _i}{h_{ki}}{a_{ki}}}}{{\beta _i^2h_{ki}^2 + {\mu _{ki}}}},\forall k = 1,2
\end{eqnarray}
From KKT conditions, when the answer is in the feasible region or for $\frac{{{a_{ki}}}}{{{\beta _i}{h_{ki}}}} < \sqrt {{P_k}}$, we have
\begin{eqnarray}
{\alpha _k} = \frac{{{a_{ki}}}}{{{\beta _i}{h_{ki}}}} ,\forall k = 1,2
\end{eqnarray}
and when the answer is on the constraint, we have ${\mu _{ki}}\ge 0$ $,\forall k = 1,2$. In this case, we can easily derive that 
\begin{eqnarray}
 {\mu _{ki}}=\frac{{{\beta _i}{h_{ki}}{a_{ki}}}}{{\sqrt {{P_k}} }} - \beta _i^2 h_{ki}^2
\end{eqnarray}
hence, we have
\begin{eqnarray}
{\alpha _k} = \sqrt {{P_k}} ,\forall k = 1,2
\end{eqnarray}
The two described steps are iterated successively until the PAV converges. The above procedures are summarized in Algorithm 1. The parameter $\delta$ used in the algorithm denotes the convergence telorance. 

In some applications, the total system power is more imprtant than the individual user power. In addition, to have a fair comparison of the proposed scheme (A-CMF) with the other ones, they should be compared under the same total system power. As a result, in the following, we consider the A-CMF scheme under the power constraint on the total powers of the users. In other words, we minimize (14) with the new constraint $\alpha _1^2+\alpha _2^2 \le {P_1}+{P_2}$. It is clear that in this case, the feasible region is larger and hence the performance should be better. Similar to the previous case (separate power constraint), the KKT conditions results to
 \begin{eqnarray}
\frac{{\partial {\varepsilon _i}}}{{\partial {\alpha _k}}} + {\mu _{i}}\frac{{\partial \alpha _k^2}}{{\partial {\alpha _k}}} = 0,\forall k = 1,2
\end{eqnarray}
where ${\mu _{i}}$ is the KKT coefficient.
This leads to
\begin{eqnarray}
{\alpha _k} = \frac{{{\beta _i}{h_{ki}}{a_{ki}}}}{{\beta _i^2h_{ki}^2 + {\mu _{i}}}},\forall k = 1,2
\end{eqnarray}
From KKT conditions, when the constraint holds with equality, ${\mu _{i}}$, obtained by solving the following equation, must be nonnegative.
\begin{eqnarray}
(\frac{{{\beta _i}{h_{1i}}{a_{1i}}}}{{\beta _i^2h_{1i}^2 + {\mu _{i}}}})^2+(\frac{{{\beta _i}{h_{2i}}{a_{2i}}}}{{\beta _i^2h_{2i}^2 + {\mu _{i}}}})^2={P_1}+{P_2}
\end{eqnarray}
This equation has not a straightforward answer and can be solved by bisection method [21]. 
\\On the other hand, when the answer is inside the feasible region, we should have
\begin{eqnarray}
(\frac{{{\beta _i}{h_{ki}}{a_{ki}}}}{{\beta _i^2h_{ki}^2 }})^2+(\frac{{{\beta _i}{h_{ki}}{a_{ki}}}}{{\beta _i^2h_{ki}^2 }})^2 < {P_1}+{P_2}
\end{eqnarray}
and the PAV is given by
\begin{eqnarray}
{\alpha _k} = \frac{{{a_{ki}}}}{{{\beta _i}{h_{ki}}}},\forall k = 1,2
\end{eqnarray}

When ${\mathbf{a}_i}$ is either ${\mathbf{e}_1}$ or ${\mathbf{e}_2}$, according to the computation rate in (12), the optimum value of power for each user is the maximum possible value, i.e. we have $\alpha_1^2=P_1$ and $\alpha_2^2=P_2$.

This algorithm is implemented in each relay. Then the relay with the highest computation rate (given in (7) and (12)) is selected as the best relay, as will be described in the following in section III.C. The PAV of the best relay is sent to the users through a feedback channel in order the users to adjust their transmission powers.

\begin{center}
\line(1,0){470}
\end{center}
\vspace{-15pt}
\textbf{Algorithm 1: calculating PAV for relay i}	
\vspace{-15pt}
\begin{center}
\line(1,0){470}
\end{center}
\vspace{-5pt}

Initialize $\alpha _k^{\left( 0 \right)}$ $,\forall k = 1,2$ and $\delta $

\uline{Iterate}

\hspace*{+8pt}1.ECV search: update $\mathbf{a}_i$ and ${\beta _i}$ from (4) and (5) for fixed $\alpha _k^{\left( j \right)} $ $,\forall k=1,2$

\hspace*{+8pt}2.Update $\alpha _k^{\left( {j + 1} \right)}$ $,\forall k = 1,2$ for fixed $\mathbf{a}_i$ and ${\beta _i}$ as follows

\hspace*{+30pt}for $\frac{{{a_{ki}}}}{{{\beta _i}{h_{ki}}}} < \sqrt {{P_k}}$ $\Rightarrow$ ${\alpha _k} = \frac{{{a_{ki}}}}{{{\beta _i}{h_{ki}}}} , \forall k = 1,2$ 

\hspace*{+30pt}and for $\frac{{{\beta _i}{h_{ki}}{a_{ki}}}}{{\sqrt {{P_k}} }} - \beta _i^2 h_{ki}^2 \ge 0$ $\Rightarrow$ ${\alpha _k} = \sqrt {{P_k}} ,\forall k = 1,2$

\uline{Until} $|\alpha _k^{\left( {j + 1} \right)} - \alpha _k^{\left( j \right)}{|^2} \le \delta$  $,\forall k$ 
\vspace{-5pt}
\begin{center}
\line(1,0){470}
\end{center}

\subsection{Recovering the message by each user}
The selected relay $m$ sends its recovered equation $s_m$, with the rate $R_m^r$, to both users. Each user by receiving this equation and having its own message can recover the other user's message. The received signal by user $i$ is written as
\begin{eqnarray}
{y_i} = {h_{im}}{s_m} + {z_i} ,i=1,2
\end{eqnarray}

Since the channel from the relay $m$ to the user $i ,i=1,2$, is a simple point to point channel, the achievable rate for the transmission of the equation $s_m$ is
\begin{eqnarray}
{R_{im}} = {\rm{log}}\left( {1 + {P_r}h_{im}^2} \right)
\end{eqnarray}
\\The rate in (27) is achievable using CMF strategy [8]. Like as (2) and (3), in one user computation case, the scaling factor ${\gamma _i}$ and the recovered equation for user $i$ are given by
\begin{eqnarray}
{\gamma _i} = \frac{{{P_r}h_{im}}}{{1 + {P_r}}h_{im}^2}
\end{eqnarray}
\begin{eqnarray}
{s_m} = Q\left( {{\gamma _i}{y_i}} \right)
\end{eqnarray}
The rate of recovering both messages in both users is easily given by
\begin{eqnarray}
{R_{{\rm{scheme}}}} = {\rm{min}}\left\{ {R_m^r,{R_{1m}},{R_{2m}}} \right\}
\end{eqnarray}

\subsection{Best Relay Selection}

From (30), to maximize the rate of recovering both messages in both users, denoted by $R_{scheme}$, the relay $m$ must be selected as
\begin{eqnarray}
m = arg\mathop {\max }\limits_i {\rm{min}}\left\{ {R_i^r,{R_{1i}},{R_{2i}}} \right\}
\end{eqnarray}

\begin{theorem}
The computation rate of the equation $s_i$ in each relay $i$, i.e. $R_i^r$, is upper bounded as
\begin{eqnarray}
R_i^r \le \log \left( {1 + {\rm{min}}\{ {\alpha _1^2}h_{1i}^2,{\alpha _2^2}h_{2i}^2\} } \right)
\end{eqnarray}
\end{theorem}

\begin{proof}
(Note: For the sake of simplicity, the subscript index $i$ of the rates is removed.)

First we consider the case that the ECV ${\mathbf{a}} \ne {\mathbf{e}_1} $ or ${\mathbf{e}_2}$ :\\
From (7), we need to show that
\begin{eqnarray}
\mathop {{\rm{max}}}\limits_{{\mathbf{a}} \ne {\mathbf{e}_1},{\mathbf{e}_2}} {\log ^ + }\left( {{{\left( {{{\left| {\left| {\mathbf{a}} \right|} \right|}^2} - \frac{{{{\left| {{{\mathbf{h}}^T}{\mathbf{a}}} \right|}^2}}}{{1 + {{\left| {\left| {\mathbf{h}} \right|} \right|}^2}}}} \right)}^{ - 1}}} \right)
\le \log \left( {1 + {\rm{min}}\{ {\alpha _1^2}h_{1}^2,{\alpha _2^2}h_{2}^2\} } \right)
\end{eqnarray}
It is sufficient to show that
\begin{eqnarray}
\frac{{1 + {{\left| {\left| {\mathbf{h}} \right|} \right|}^2}}}{{\mathop {{\rm{min}}}\limits_{{\mathbf{a}} \ne {\mathbf{e}_1},{\mathbf{e}_2}} {{\left| {\left| {\mathbf{a}} \right|} \right|}^2} + {{\left| {\left| {\mathbf{h}} \right|} \right|}^2}{{\left| {\left| {\mathbf{a}} \right|} \right|}^2} - {{\left| {{{\mathbf{h}}^T}{\mathbf{a}}} \right|}^2}}}
\le 1 + {\rm{min}}\left( {{\alpha _1^2}h_1^2,{\alpha _2^2}h_2^2} \right)
\end{eqnarray}
Which can be written as
\begin{eqnarray}
1 + {\left| {\left| {\mathbf{h}} \right|} \right|^2} \le \left( {1 + {\rm{min}}\left( {{\alpha _1^2}h_1^2,{\alpha _2^2}h_2^2} \right)} \right)
\times \mathop {{\rm{min}}}\limits_{{\mathbf{a}} \ne {\mathbf{e}_1},{\mathbf{e}_2}} \{ {\left| {\left| {\mathbf{a}} \right|} \right|^2} + {\left| {\left| {\mathbf{h}} \right|} \right|^2}{\left| {\left| {\mathbf{a}} \right|} \right|^2} - {\left| {{{\mathbf{h}}^T}{\mathbf{a}}} \right|^2}\} 
\end{eqnarray}
where $\mathbf{h}= \left[ {\begin{array}{*{20}{c}}{{\alpha _1}{h_1}}\\{{\alpha _2} {h_2}}\end{array}} \right]$ and ${\mathbf{a}} = \left[ {\begin{array}{*{20}{c}}{{a_1}}\\{{a_2}}\end{array}} \right]$.
We define the variable $J$ as
$$J \buildrel \Delta \over = {\left| {\left| {\mathbf{a}} \right|} \right|^2} + {\left| {\left| {\mathbf{h}} \right|} \right|^2}{\left| {\left| {\mathbf{a}} \right|} \right|^2} - {\left| {{{\mathbf{h}}^T}{\mathbf{a}}} \right|^2}$$
which with some straightforward simplifications and by defining ${g_1} \buildrel \Delta \over = {\alpha _1}{h_1}$ and ${g_2} \buildrel \Delta \over = {\alpha _2}{h_2}$, results in
\begin{eqnarray}
J = a_1^2 + a_2^2 + \left( {{g_1}{a_2} - {g_2}{a_1}} \right)^2
\end{eqnarray}
Without loss of generality, we assume that ${g_1} \le {g_2}$.
Then
\begin{eqnarray}
\mathop {{\rm{min}}}\limits_{{\mathbf{a}} \in {\mathbb{Z}^2},{\mathbf{a}} \ne {\mathbf{e}_1},{\mathbf{e}_2}} J=\mathop {{\rm{min}}}\limits_{a_1,a_2 \in {\mathbb{Z}},a_1,a_2\ge{1}} J\ge{\mathop {{\rm{min}}}\limits_{a_1\in {\mathbb{Z}},a _2\in{\mathbb{R}},a_1\ge{1}} J}
\end{eqnarray}
For the right side of the inequality, we can minimize $J$ with respect to $a_2$ as follows
\begin{eqnarray}
\frac{{\partial J}}{{\partial {a_2}}} = 2{a_2} + 2{g_1}\left( {{g_1}{a_2} - {g_2}{a_1}} \right) = 0
\end{eqnarray}
which leads to
\begin{eqnarray}
{a_2} = \frac{{{g_1}{g_2}}}{{1 + g_1^2}}{a_1}
\end{eqnarray}
By substituting (39) in (36) and some straightforward simplifications, we obtain
\begin{eqnarray}
J = a_1^2(1 + g_2^2 - (g_1^2g_2^2)/(1 + g_1^2))
\end{eqnarray}
Thus, from (37) we get
\begin{eqnarray}
\mathop {{\rm{min}}}\limits_{{\mathbf{a}} \ne {\mathbf{e}_1},{\mathbf{e}_2}} J \ge \mathop {{\rm{min}}}\limits_{{a_1} \ge 1,{a_1} \in \mathbb{Z}} a_1^2\left( {1 + g_2^2 - \frac{{g_1^2g_2^2}}{{1 + g_1^2}}} \right)
\end{eqnarray}
Since we have
\begin{eqnarray}
\mathop {{\rm{min}}}\limits_{{a_1} \ge 1,{a_1} \in \mathbb{Z}} a_1^2\left( {1 + g_2^2 - \frac{{g_1^2g_2^2}}{{1 + g_1^2}}} \right) = 1 + g_2^2 - \frac{{g_1^2g_2^2}}{{1 + g_1^2}}
\end{eqnarray}
The right-hand side of (35) can be written as
\begin{eqnarray}
\left( {1 + {\rm{min}}(g_1^2,g_2^2} \right))\mathop {{\rm{min}}}\limits_{{\mathbf{a}} \ne {\mathbf{e}_1},{\mathbf{e}_2}} J \ge  
\left( {1 + g_1^2} \right)\left ( {1 + g_2^2 - \frac{{g_1^2g_2^2}}{{1 + g_1^2}}} \right) =1 + g_1^2 + g_2^2 
\end{eqnarray}
This proves (35), and then (33).
 
Now, we consider the case that $\mathbf{a}={\mathbf{e}_1}$ or ${\mathbf{e}_2}$ . It is clear that
\begin{eqnarray}
{R^r} = \min \left( {{R_e},R_{{e^c}}} \right) \le R_{{e^c}} 
\end{eqnarray}
Hence, the theorem is proved.
\end{proof}

The bound derived is tight, specially at Low SNRs. Please note that at low SNRs, the ECV as a solution of (4) is usually either ${\mathbf{e}_1}$ or ${\mathbf{e}_2}$ [15]. In this case, from (36) $J$ can be written as
\begin{eqnarray}
J_0 = 1 + g_1^2
\end{eqnarray}
which certainly is lower than $\mathop {{\rm{min}}}\limits_{{\mathbf{a}} \in {\mathbb{Z}^2},{\mathbf{a}} \ne {\mathbf{e}_1},{\mathbf{e}_2}} J$, (please note that we have assumed ${\mathbf{e}_k}$ as the solution). For $J_0$ lower than the bound, i.e. ${\mathop {{\rm{min}}}\limits_{a_1\in {\mathbb{Z}},a _2\in{\mathbb{R}},a_1\ge{1}} J}$, we have
 \begin{eqnarray}
1 + g_1^2\le 1 + g_2^2 - \frac{{g_1^2g_2^2}}{{1 + g_1^2}}=1 + \frac{{g_2^2}}{{1 + g_1^2}}
\end{eqnarray}
that leads to $R_{{e^c}} \le R_e$ and ${R^r} = R_{e^c}$.
Hence, at low SNRs, the rate is very close to the bound given in (32) with a high probability.

According to this theorem, we have
\begin{eqnarray}
R_i^r \le \log \left( {1 + {\rm{min}}\{ {P_1}h_{1i}^2,{P_2}h_{2i}^2\} } \right)
\end{eqnarray}
Using (47) and with the assumation of ${P_r} \ge {\rm{max}}\left( {{P_1},{P_2}} \right)$ we can easily rewrite (31) as
\begin{eqnarray}
m = arg\mathop {\max }\limits_i R_i^r
\end{eqnarray}
Now, the best relay, after the multiple access phase, is selected based on (48), using the approach similar to [22]. That is the relay $i, i=1,2,...,M$, sets a timer with the value $T_i$ proportional to the inverse of its corresponding rate, i.e. $R_i^r$. The first relay that its timer reaches zero (which has the highest rate) broadcasts a flag, to inform other relays, and is selected as the best relay for the broadcast phase.

\section{PERFORMANCE ANALYSIS}
From (30), the outage probability of the proposed scheme can be computed as
\begin{eqnarray}
P_{out}^{CMF}(R_{t}) &=& pr\left( {{R_{{\rm{scheme}}}} < {R_t}} \right) \nonumber\\
 &=&pr\left( {\min \left\{ {R_m^r,{R_{1m}},{R_{2m}}} \right\} < {R_t}} \right)
\end{eqnarray}
where $R_t$ denotes the target rate. According to the Theorem 1 and with the assumation of ${P_r} \ge {\rm{max}}\left( {{P_1},{P_2}} \right)$, we have
\begin{eqnarray}
P_{out}^{CMF}(R_{t}) = pr(R_m^r < {R_t})
\end{eqnarray}
Moreover, from Theorem 1, a lower bound for the outage probability is derived as follows
\begin{eqnarray}
P_{out}^{CMF,bound}(R_{t}) &=&pr({\rm{ma}}{{\rm{x}}_{i = 1, \ldots ,M}}\log \left( {1 + {\rm{min}}\{ {P_1}h_{1i}^2,{P_2}h_{2i}^2\} } \right) < {R_t}) \nonumber\\
&=&pr({\rm{ma}}{{\rm{x}}_{i = 1, \ldots ,M}}\min \left( {{P_1}h_{1i}^2,{P_2}h_{2i}^2} \right) < {2^{{R_t}}} - 1)
\end{eqnarray}
With the defination of ${\gamma _i} \buildrel \Delta \over = {\rm{min}}\left\{ {{P_1}h_{1i}^2,{P_2}h_{2i}^2} \right\}$ and since ${\gamma _i}$s are independent exponential random variables with the following CDF:
\begin{eqnarray}
{F_{{\gamma _i}}}\left( \gamma  \right) = 1 - {e^{ - \left( {\frac{1}{{{P_1}\sigma _{1i}^2}} + \frac{1}{{{P_2}\sigma _{2i}^2}}} \right)\gamma }}
\end{eqnarray}
we can find the CDF of ${\gamma _{max}} \buildrel \Delta \over = \mathop {\max }\limits_i {\gamma _i}$ as
\begin{eqnarray}
pr\left( {{\gamma _{max}} < \gamma } \right) = \mathop \prod \limits_{i = 1}^M pr({\gamma _i} < \gamma )= \mathop \prod \limits_{i = 1}^M {F_{{\gamma _i}}}\left( \gamma  \right)
\end{eqnarray}
Hence, the outage probability lower bound can be easily computed as
\begin{eqnarray}
P_{out}^{CMF,bound}(R_{t})=\mathop \prod \limits_{i = 1}^M 1 - {e^{ - \left( {\frac{1}{{{P_1}\sigma _{1i}^2}} + \frac{1}{{{P_2}\sigma _{2i}^2}}} \right)\left( {{2^{{R_t}}} - 1} \right)}}
\end{eqnarray}
From Taylor series expansion, in high SNRs, when $P_1=P_2=P$, we can approximate (54) as
\begin{eqnarray}
P_{out}^{CMF,bound}(R_{t}) \le \mathop \prod \limits_{i = 1}^M \left( {\frac{{{2^{{R_t}}} - 1}}{P}} \right)\left( {\frac{1}{{\sigma _{1i}^2}} + \frac{1}{{\sigma _{2i}^2}}} \right)=\frac{{{{({2^{{R_t}}} - 1)}^M}}}{{{P^M}}}\mathop \prod \limits_{i = 1}^M \left( {\frac{1}{{\sigma _{1i}^2}} + \frac{1}{{\sigma _{2i}^2}}} \right)
\end{eqnarray}
Hence, the acheivable diversity order from the outage bound with the definition $G =  - \mathop {\lim }\limits_{P \to \infty } \frac{{{\rm{log}}({P_{out}})}}{{{\rm{log}}\left( P \right)}}$ [23] is equal to M, i.e. the number of relays.

According to the Theorem1 and with the assumation of ${P_r} \ge {\rm{max}}\left( {{P_1},{P_2}} \right)$, an upperbound on the sum rate conditioned  on each channel realizations can be derived as
\begin{eqnarray}
R_{sum}^{CMF,bound}(h_{1},h_{2}) = 2 \log \left( {1 + {\rm{ma}}{{\rm{x}}_{i = 1, \ldots ,M}}\min \left( {{P_1}h_{1i}^2,{P_2}h_{2i}^2} \right) } \right) = 2{\rm{log}}\left( {1 + {\gamma _{max}}} \right)
\end{eqnarray}
The unconditional sum rate can be computed by taking the expectation of (56) as
\begin{eqnarray}
R_{sum}^{CMF,bound}=\mathop \smallint \limits_0^\infty  2\log ( {1 + \gamma } ){f_{{\gamma _{max}}}}( \gamma )d{\gamma} 
\end{eqnarray}
where $f_{\gamma _{max}}$ is the PDF of $\gamma _{max}$, which can be easily obtained from its CDF given in (52)-(53). With some straightforward simplifications, leads to
\begin{eqnarray}
R_{sum}^{CMF,bound}&=&\mathop \sum \limits_{i = 1}^M \left( {\frac{1}{{{P_1}\sigma _{1i}^2}} + \frac{1}{{{P_2}\sigma _{2i}^2}}} \right)\{ {I_1}\left( {\frac{1}{{{P_1}\sigma _{1i}^2}} + \frac{1}{{{P_2}\sigma _{2i}^2}}} \right)\nonumber\\\ &-& \mathop \sum \limits_{k = 1,k \ne i}^M {I_1}\left( {\frac{1}{{{P_1}\sigma _{1i}^2}} + \frac{1}{{{P_2}\sigma _{2i}^2}} + \frac{1}{{{P_1}\sigma _{1k}^2}} + \frac{1}{{{P_2}\sigma _{2k}^2}}} \right)\nonumber\\&+&\mathop \sum \limits_{k = 1,k \ne i}^{M-1} \mathop \sum \limits_{l = 1,l \ne k,i}^M {I_1}(\frac{1}{{{P_1}\sigma _{1i}^2}} + \frac{1}{{{P_2}\sigma _{2i}^2}} + \frac{1}{{{P_1}\sigma _{1k}^2}} + \frac{1}{{{P_2}\sigma _{2k}^2}} + \frac{1}{{{P_1}\sigma _{1l}^2}} + \frac{1}{{{P_2}\sigma _{2l}^2}})\nonumber\\\ &-&  \ldots  + {\left( { - 1} \right)^{M - 1}}{I_1}\left( {\mathop \sum \limits_{k = 1}^M \frac{1}{{{P_1}\sigma _{1k}^2}} + \frac{1}{{{P_2}\sigma _{2k}^2}}} \right)\} 
\end{eqnarray}
where ${I_n}\left( \mu  \right) = \mathop \smallint \limits_0^\infty  {t^{n - 1}}\ln \left( {1 + t} \right){e^{ - \mu t}}dt= \left( {n - 1} \right)!{e^\mu }\mathop \sum \limits_{l = 1}^n \frac{{\Gamma \left( {l - n,\mu } \right)}}{{{\mu ^l}}}$ and $\Gamma \left( {.,.} \right)$ is the upper incomplete gamma function defined in [24].
\section{NUMERICAL RESULTS}
For numerical evaluation, target rate $R_t=1$ is considered. The Rayleigh channel parameters equal to  $\sigma _{ji}^2 = 1,j=1,2,i=1,...,M$, are assumed. The parameter $\delta$ in algorithm 1 is setteled as ${10^{ - 3}}$.

In Fig. 2, the outage probability of the proposed schemes along with the derived lower bound given in (54), versus SNR, is plotted for $M=1,2,3$ relays and for equal maximum transmit powers for the users and the relay. As observed, for both M-CMF and A-CMF schemes, the derived lower bound is quite tight especially at high SNRs. Moreover, as expected, by the increase of the number of relays, the outage performance as well as the diversity order improve significantly. It is observed that the proposed schemes provides a diversity order of $M$, i.e., the number of relays employed.

In Fig. 3, the average sum rates of the proposed schemes along with the derived upper bound in (58) are plotted for $M = 1,2$ and for equal maximum powers for the users. As observed, the M-CMF reaches the bound only in low SNRs, while the A-CMF approaches the bound in all SNR values. In other words, A-CMF outperforms the M-CMF in high SNR at the cost of using feedback transmission.

Fig. 4 compares the symbol error rate (SER) of the proposed schemes with the ones introduced in [2], including AF and two-step DF, for $M=1$ with BPSK modulation and for equal maximum transmit powers for the users. Our proposed scheme indiciates significantly better performance about 6dB in SER equal to 0.02. Please note that the SER has been evaluated by simulation, as it is not easy at all to analytically derive the SER when using the CMF based strategy, due to an integer optimization problem being solved numerically within this strategy.

Fig. 5 compares the outage probability of the proposed schemes with the conventional strategies and also three-step DF [2], for $M=2$ and for equal maximum transmit powers for the users. The same relay selection strategy is used for all schemes. Although all of the methods provide the same order of diversity, our proposed schemes demonstrates a better performance about 2dB in high SNR values. 

Fig. 6 compares the average sum rate of the proposed schemes with the conventional strategies, for $M = 2$ and for equal transmit powers. As it is observed, our proposed schemes perform significantly better than the conventional strategies in all SNRs. For example, in sum rate 4, A-CMF has 4dB and M-CMF has 2dB improvement in comparison with the best conventional relaying scheme.

Fig. 7 compares the average sum rate of the A-CMF scheme under two different power constraints, one in each user power and the other on the total power, for $M=1,2$ relays. For the first case, maximum transmission power of each user is considered to be equal to $P$, i.e. $P_{1}=P_{2}=P$, while for the second case, the maximum total transmission powers of both users is considered to be equal to $2P$. As it is observed from this figure, in the latter case, the system has a better performance. The reason is that in the second case, the feasible region of the optimization in (14) is larger than the one of the first case, that results in a higher rate. When comparing with the other scheme, the latter constraint is more reasonable, as different schemes should be compared under the same total transmission powers. Since our results above indicate that the A-CMF under the maximum power constraint on each user transmission performs better than the other schemes, specially in term of the average sum rate, we didn't bring the comparison results with the other schemes, when the A-CMF is designed under the constraint on the total transmission power.

In Fig. 8, we evalute the perfromance of proposed schemes along with the conventional strategies for $M=2$ when all link are modeled as complex gaussian channels with variance one. The figure shows that the outage probability of the peroposed schemes, specially A-CMF, are better than those of the conventional strategies. 

In Fig. 9, the performance of proposed scheme has been evaluated when the channels' variances are not identical. This Fig. shows the outage probability of proposed schemes and two-step DF versus SNR for different values of ${\rm{delta}} \buildrel \Delta \over = |\sigma _1^2 - \sigma _2^2|$, where $\sigma _{ji}^2 = \sigma _j^2,j=1,2,i=1,...,M$. delta in fact indicates the difference between the two users' channel variances. In this Fig., we have $M=2$. For a fair camparison, the sum of the two channels' variances is set equal to two, i.e. $\sigma _1^2 + \sigma _2^2 = 2$. As can be observed, the lower delta makes better perfromance, however diversity order does not change with delta. From this Fig., the perfromance of the proposed schemes are better than two-step DF, which shows the best performance among the conventional schemes (please see Fig. 5). As expected, the amount of the improvement decreases by the increase of the delta. For example, in outage $10^{-2}$, while at delta equal to 0.5, the proposed schemes have 1.8dB better perfromance than the two-step DF, the improvement is 1.4dB at delta equal to 1.

\section{CONCLUSION}
In this paper, based on CMF strategy, a novel two-way relaying scheme, for two cases of relays with and without capability of feedback transmission, is proposed that improves the network throughput significantly. Furthermore, a relay selection scheme is exploited to achieve a higher order of diversity through employing multiple relays. By theoretical analysis, an upper bound on the computation rate of each relay is derived and based on that, a tight lower bound on the outage probability and an upper bound on the average sum rate of the system are presented. Our numerical results showed that the proposed scheme, in both cases of with and without using feedback, performs significantly better than the AF and DF strategies in terms of the outage probability, average sum rate, and the symbol error rate, and also provides a diversity order equal to the number of relays employed.

\newpage

\begin{figure}[tb!]
\centering

\includegraphics[width =4.5in]{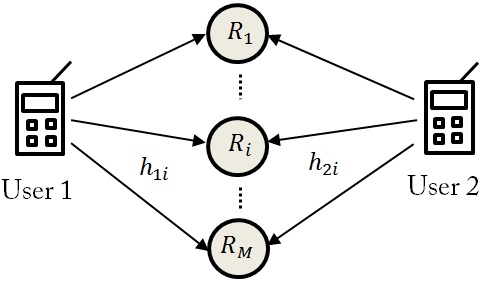}

\caption{Two way relay channel}

\end{figure}

\newpage
~
\newpage

\begin{figure}[h]
\centering
\center
\includegraphics[width =\columnwidth]{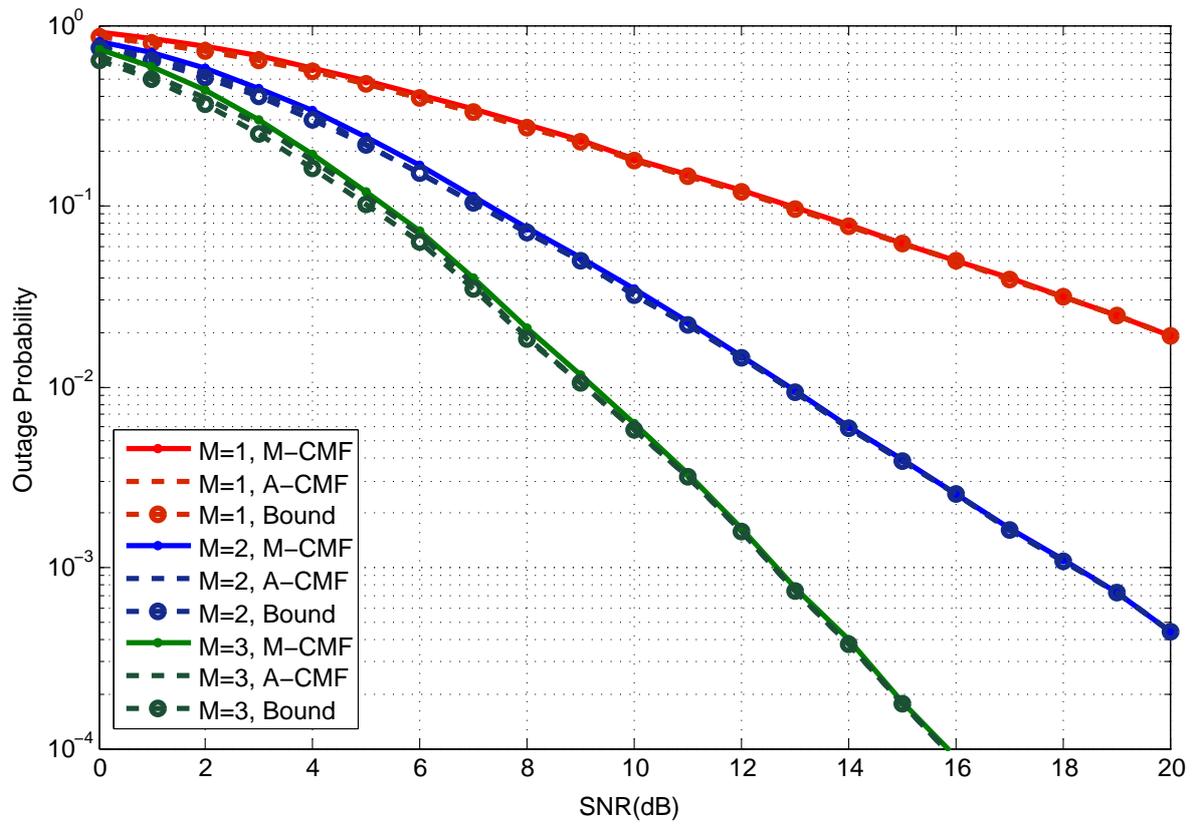}
\caption{Outage probability of the proposed schemes along with the derived lower bound versus SNR (M=1,2,3).}

\end{figure}

\newpage

\begin{figure}[h]
\centering
\center
\includegraphics[width =\columnwidth]{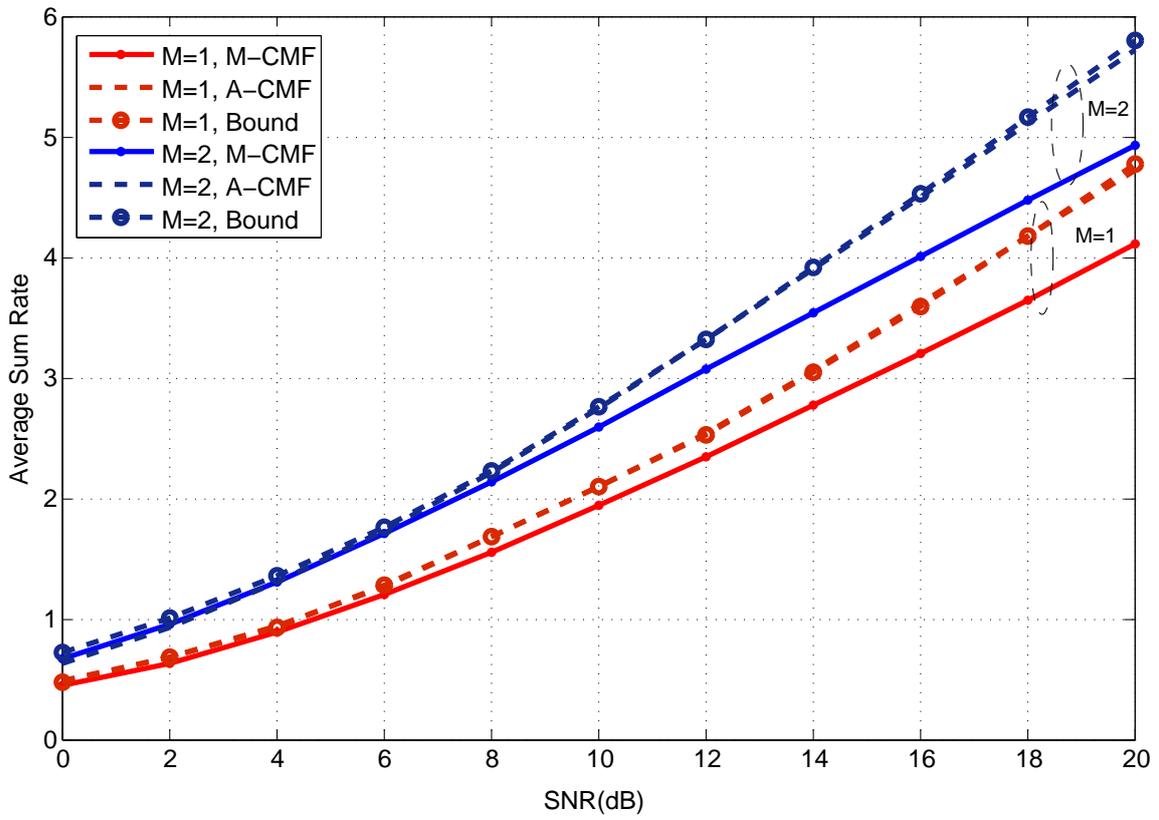}
\caption{Average sum rate of the proposed schemes along with the proposed upper bound versus SNR (M=1,2).}

\end{figure}

\newpage

\begin{figure}[h]
\centering
\center
\includegraphics[width =\columnwidth]{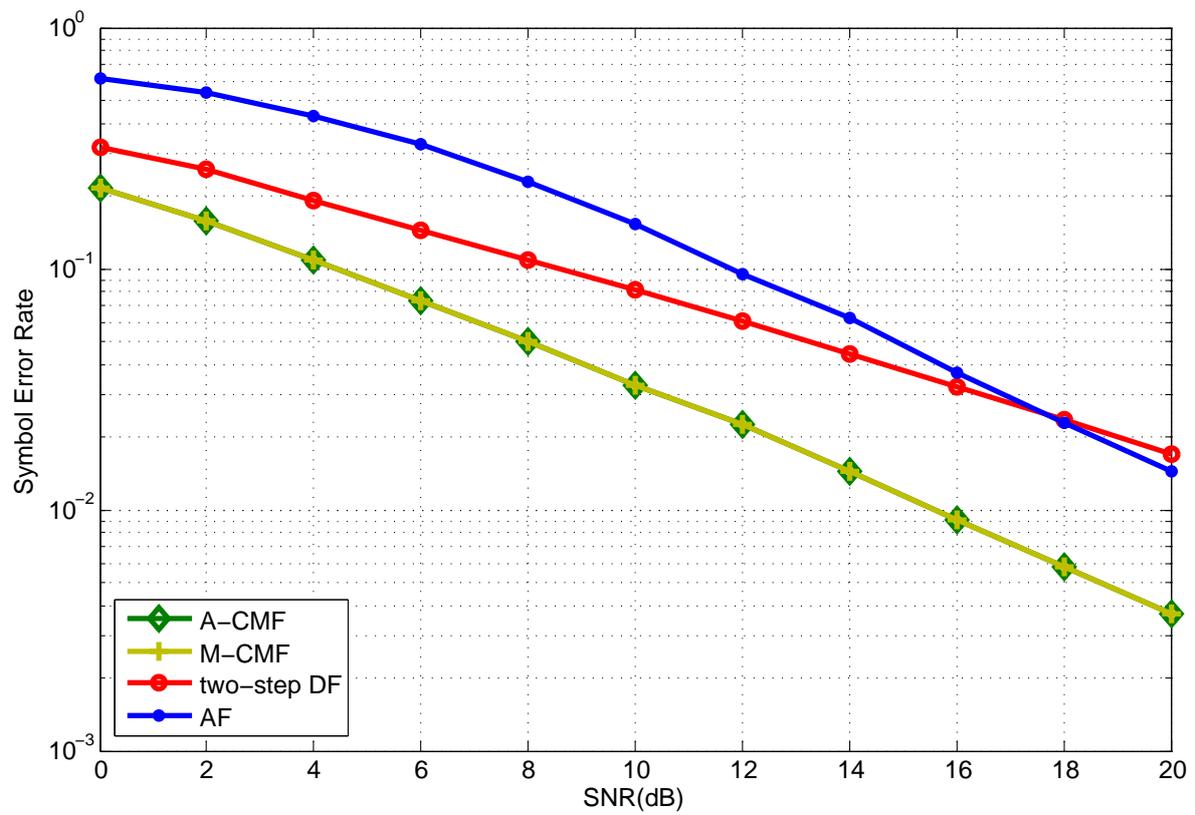}
\caption{Symbol error rate of the proposed schemes in comparison with conventional strategies versus SNR (M=1).}

\end{figure}

\newpage

\begin{figure}[h]
\centering
\center
\includegraphics[width =\columnwidth]{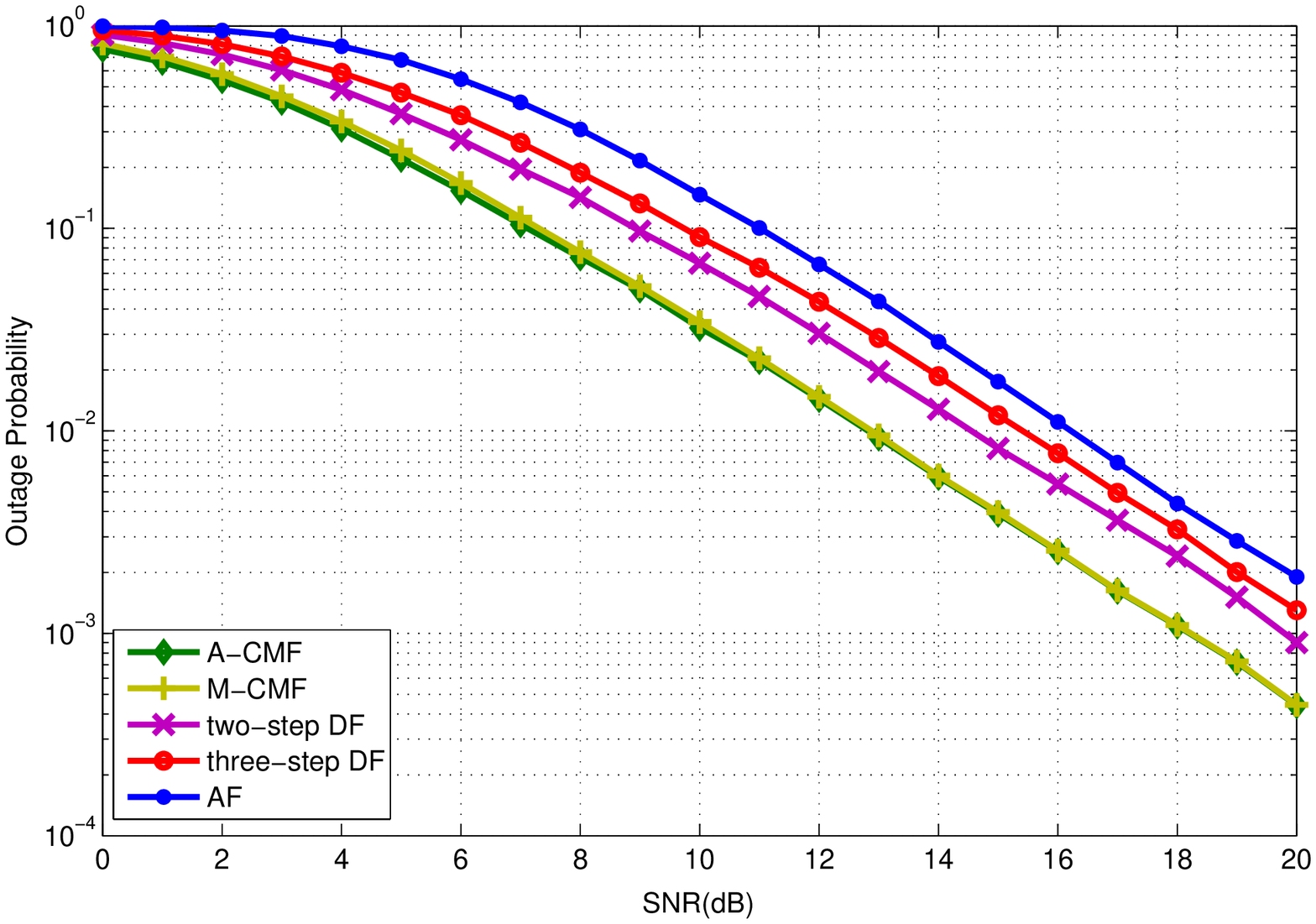}
\caption{Outage probability of the proposed schemes in comparison with conventional strategies versus SNR (M=2).}

\end{figure}

\newpage

\begin{figure}[h]
\centering
\center
\includegraphics[width =\columnwidth]{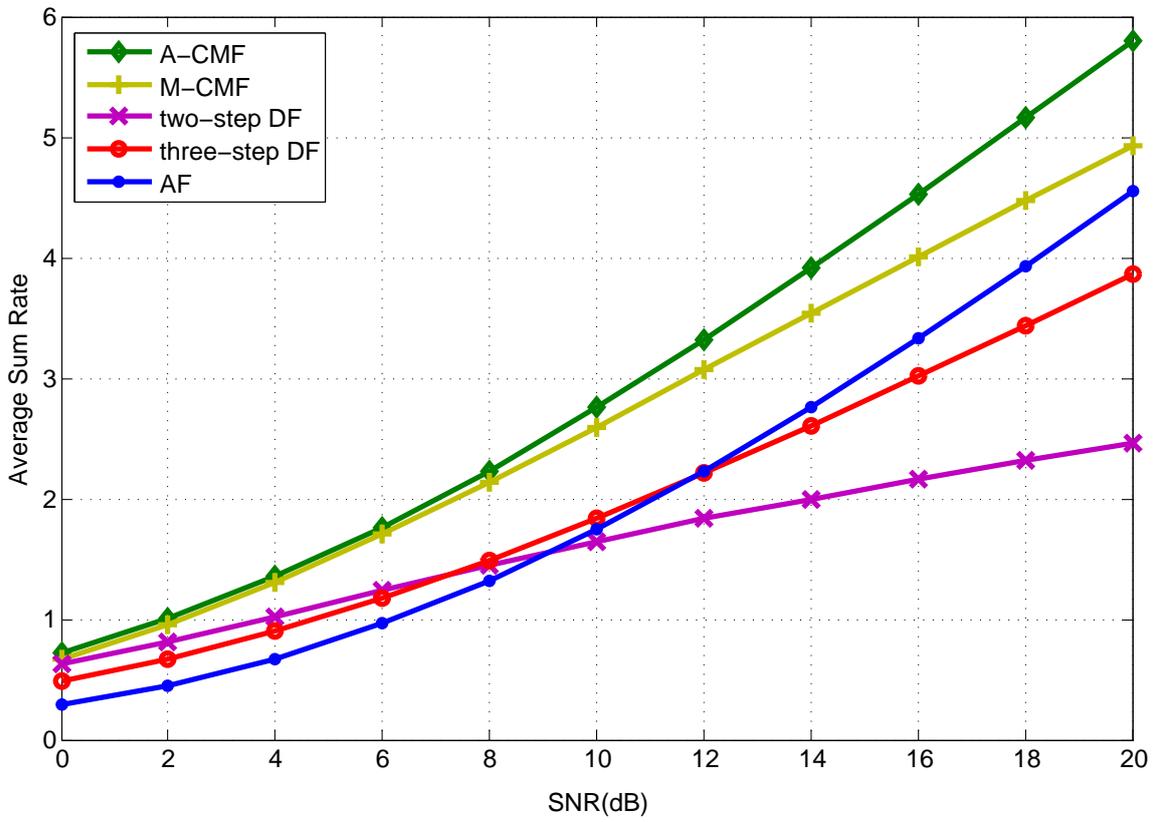}
\caption{Average sum rate of the proposed schemes in comparison with conventional strategies versus SNR (M=2).}

\end{figure}

\newpage

\begin{figure}[h]
\centering
\center
\includegraphics[width =\columnwidth]{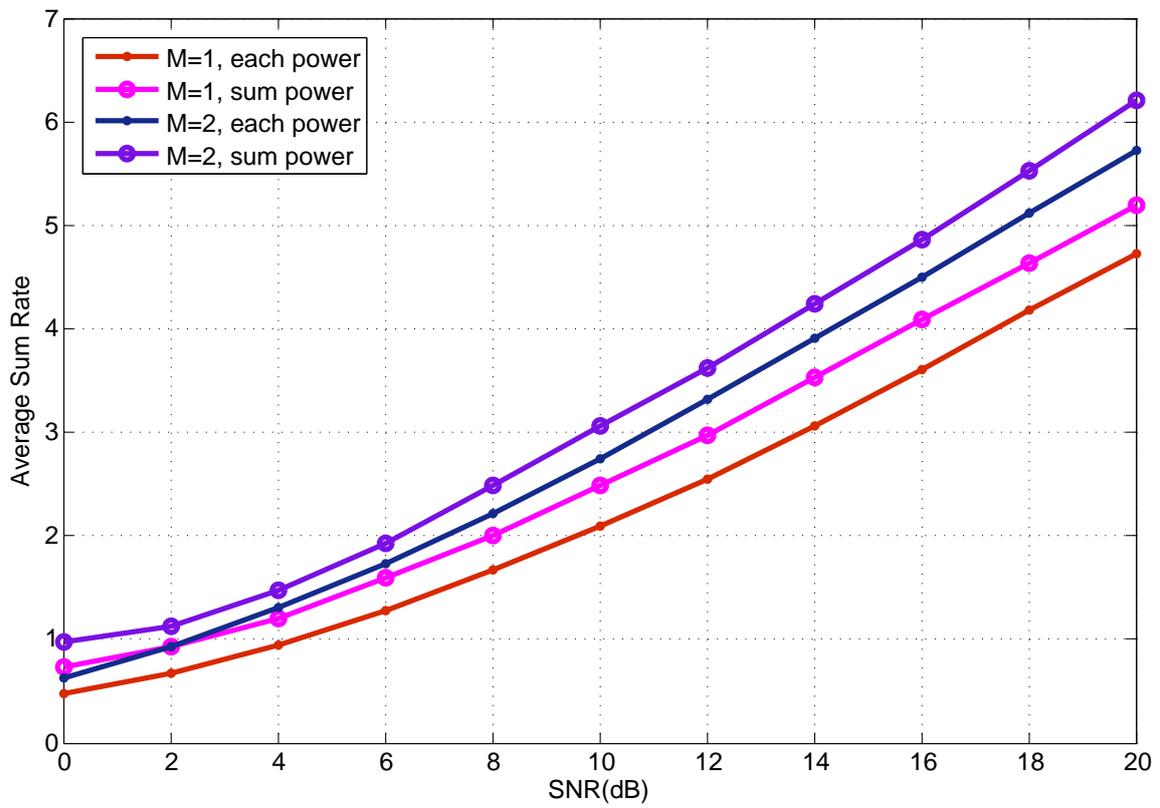}
\caption{Average sum rate of the A-CMF in the cases of constraint on the user's sum power or each power versus SNR (M=1,2).}
\end{figure}
\newpage

\begin{figure}[h]
\centering
\center
\includegraphics[width =\columnwidth]{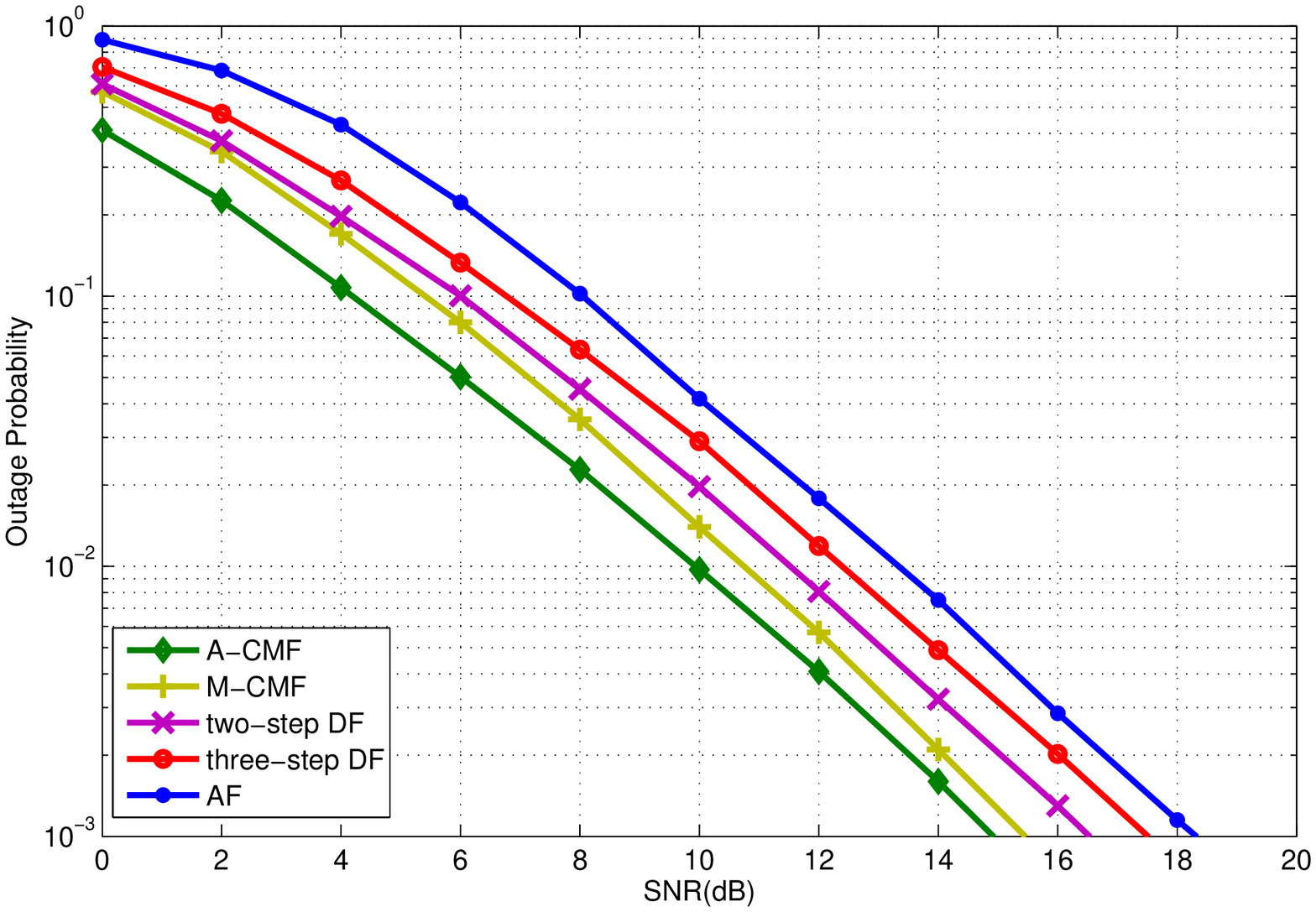}
\caption{Outage probability of the proposed schemes in comparison with the conventional strategies versus SNR for complex Gaussian Channels (M=2).}
\end{figure}
\newpage

\begin{figure}[h]
\centering
\center
\includegraphics[width =\columnwidth]{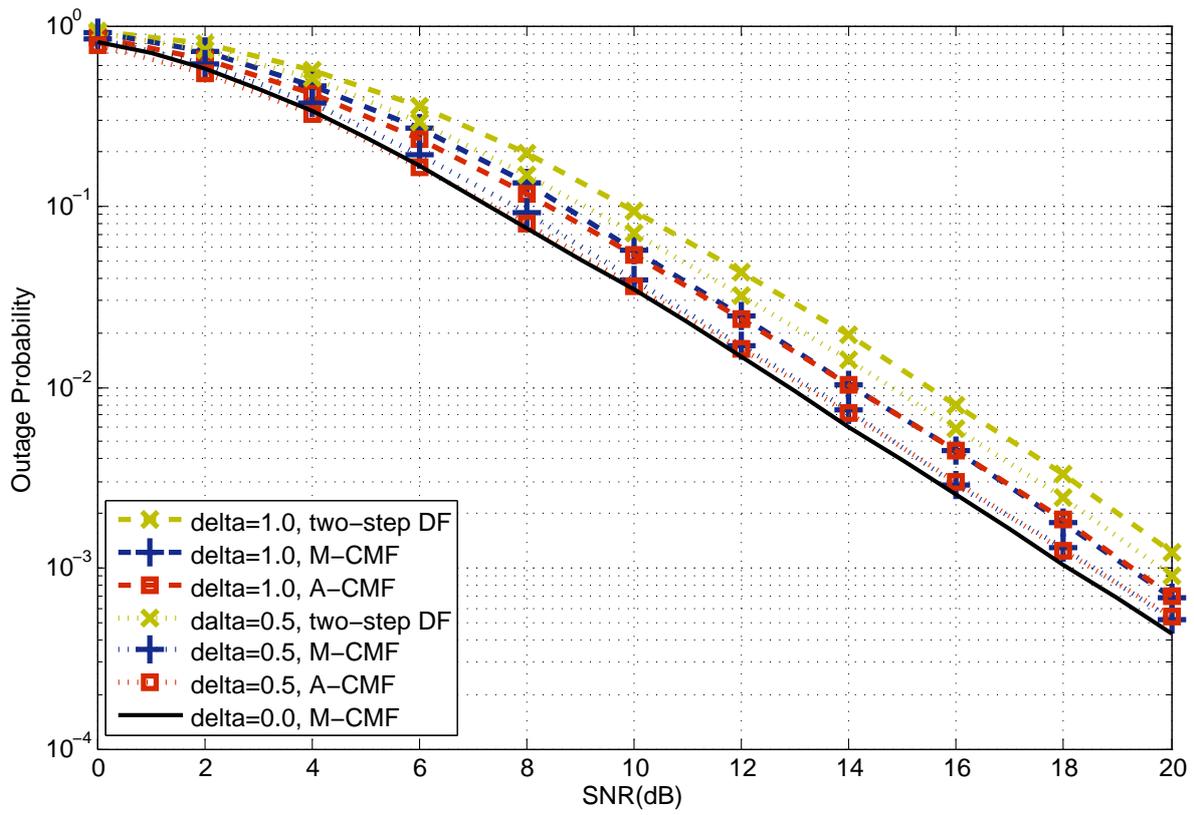}
\caption{Outage probability of the proposed schemes and two-step DF, for different values of delta, (M=2).}
\end{figure}
\newpage



\end{document}